\documentclass[journal]{IEEEtran}

\usepackage[cmex10]{amsmath}
\usepackage{amssymb, amsfonts, amsthm}
\usepackage{calrsfs}
\usepackage{graphicx}

\newcommand{\CONF}{\textnormal{CONF}}

\renewcommand{\a}{\mathbf{a}}
\renewcommand{\b}{\mathbf{b}}

\newcommand{\q}{\mathbf{q}}
\newcommand{\s}{\mathbf{s}}

\newcommand{\x}{\mathbf{x}}
\newcommand{\y}{\mathbf{y}}
\newcommand{\z}{\mathbf{z}}

\newcommand{\A}{\mathcal{A}}
\newcommand{\B}{\mathcal{B}}
\newcommand{\C}{\mathcal{C}}

\renewcommand{\H}{\mathcal{H}}
\renewcommand{\P}{\mathcal{P}}
\newcommand{\Q}{\mathcal{Q}}
\newcommand{\R}{\mathcal{R}}

\newcommand{\W }{\mathcal{W}}
\newcommand{\WQ }{\overline{\mathcal{W}}}
\newcommand{\X}{\mathcal{X}}
\newcommand{\Y}{\mathcal{Y}}
\newcommand{\Z}{\mathcal{Z}}
\newcommand{\U}{\mathcal{U}}
\renewcommand{\S}{\mathcal{S}}
\newcommand{\SB}{\P(\S)}

\newcommand{\EE}{\mathbb{E}}
\newcommand{\PP}{\mathbb{P}}

\newcommand{\eps}{\varepsilon}

\newtheorem{thm}{Theorem}
\newtheorem{lem}[thm]{Lemma}

\theoremstyle{definition}
\newtheorem{defn}[thm]{Definition}
\theoremstyle{remark}
\newtheorem{rem}{Remark}
\newtheorem{ex}{Example}

\title{The Arbitrarily Varying Multiple-Access Channel with Conferencing Encoders}
\author{Moritz~Wiese,~\IEEEmembership{Student~Member,~IEEE,}
        and Holger~Boche,~\IEEEmembership{Fellow,~IEEE}
\thanks{M. Wiese and H. Boche are with the Lehrstuhl f\"{u}r Theoretische Informationstechnik, Technische Universit\"{a}t M\"unchen, Munich, Germany (e-mail: \{wiese,boche\}@tum.de)}
\thanks{The results of this paper were presented at the 2011 IEEE International Symposium on Information Theory (ISIT '11), St. Petersburg, Russia.}
\thanks{This work was supported by the German Ministry of Education and Research (BMBF) under
Grant 01BQ1050 and the DFG COIN Project.}
\thanks{Copyright (c) 2012 IEEE. Personal use of this material is permitted.  However, permission to use this material for any other purposes must be obtained from the IEEE by sending a request to pubs-permissions@ieee.org.}}

\begin{document}

\maketitle

\begin{abstract}
  We derive the capacity region of arbitrarily varying multiple-access channels with conferencing encoders for both deterministic and random coding. For a complete description  it is sufficient that one conferencing capacity is positive. We obtain a dichotomy: either the channel's deterministic capacity region is zero or it equals the two-dimensional random coding region. We determine exactly when either case holds. We also discuss the benefits of conferencing. We give the example of an AV-MAC which does not achieve any non-zero rate pair without encoder cooperation, but the two-dimensional random coding capacity region if conferencing is possible. Unlike compound multiple-access channels, arbitrarily varying multiple-access channels may exhibit a discontinuous increase of the capacity region when conferencing in at least one direction is enabled. 
\end{abstract}


\section{Introduction}

Multiple-Access Channels (MACs) and similar multi-sender channels with conferencing encoders have attracted attention recently due to the inclusion of base-station cooperation methods in standards for future wireless systems \cite{BLW, MYK, SGPGS, Wig}. The original conferencing protocol for the discrete memoryless MAC is due to Willems \cite{Wi1,Wi2}. The conferencing MAC with imperfect channel state information was modeled as a compound MAC with conferencing encoders and considered in \cite{WBBJ11}, a different model for channel state uncertainty is given in \cite{PSSB}.

This paper covers a very high degree of channel uncertainty in MACs: the channel states may vary arbitrarily over time. The task is to use coding to enable reliable communication for every possible state sequence. The corresponding information-theoretic channel model is the Arbitrarily Varying MAC (AV-MAC). The random coding capacity region of the AV-MAC without encoder cooperation was determined in \cite{J}. Building on this result, the deterministic coding capacity region of some AV-MACs without cooperation was determined in \cite{AC}. In general, it is still open. We will use the ``robustification'' and ``elimination of correlation'' techniques developed by Ahlswede in \cite{A1,A2}, and partly already used in \cite{J} in a multi-user setting, in order to characterize both the deterministic and random coding capacity regions of any AV-MAC with conferencing encoders, i.e.\ of any AV-MAC where encoding is done using a Willems conference as in \cite{Wi1,Wi2} with at least one positive conferencing capacity. 
Thus none of the techniques we apply in this paper is completely new, but in contrast to the non-conferencing situation, they allow for the complete solution of the problems considered here. The rather general ``robustification'' technique establishes the random coding capacity region of the AV-MAC with conferencing encoders. Both single- and multi-user arbitrarily varying channels are special in that random coding as commonly used in information theory does not yield the same results as deterministic coding. This shows that common randomness shared at the senders and the receiver is an important additional resource. There is a dichotomy: either reliable communication at any non-zero rate pair is impossible with the application of deterministic codes, or the deterministic capacity region coincides with the random coding capacity region, which then is two-dimensional. In the latter case, one needs the non-standard ``elimination of correlation'' \cite{A1} for derandomization. It is a two-step protocol which 
achieves the random coding capacity region if this is possible.

The combination of the elimination technique with conferencing proves to be very fruitful. Here lies the main difference between the AV-MAC with and without conferencing. One can show that there exist channels which only achieve the zero rate pair without transmitter cooperation, but where derandomization using the elimination technique is possible if the transmitters may have a conference. The reason for this is symmetrizability. This can be interpreted in terms of an adversary knowing the channel input symbols and randomizing over the channel states. There are three kinds of symmetrizability for multiple-access channels. The capacity region of the AV-MAC without conferencing equals $\{(0,0)\}$ if all three symmetrizability conditions are satisfied. In contrast, the elimination of correlation technique works if the AV-MAC with Willems conferencing encoders does not satisfy the conditions for the first of the three kinds of symmetrizabilities. The two others do not matter. By conferencing with at least one 
positive conferencing capacity, the AV-MAC gets closer to a single-sender arbitrarily varying channel where only one symmetrizability condition exists \cite{CN}. This induced change of the channel structure is also reflected in the counter-intuitive fact that conferencing with rates tending to zero in blocklength can enlarge the capacity region. The adversary interpretation of symmetrizability highlights the importance of the AV-MAC for the theory of information-theoretic secrecy: if a channel is symmetrizable, an adversary can completely prevent communication.

The paper is organized as follows: the next section is devoted to the formalization of the channel model and the coding problems. We present the main theorems and several auxiliary coding results. The direct parts of the random and deterministic coding theorems are solved in Section \ref{sect:ach}. Section \ref{sect:conv} gives the converses of the random and deterministic AV-MAC coding theorems. Section \ref{sect:concl} concludes the paper with a discussion. In particular, the gains of conferencing are analyzed there.

\textit{Notation:} In the information-theoretic setting, we also use the terms ``encoders'' for the senders and ``decoder'' for the receiver. For any positive integer $m$, we write $[1,m]$ for the set $\{1,\ldots,m\}$. For a set $A\subset\X$, we denote its complement by $A^c:=\X\setminus A$. For real numbers $x$ and $y$, we set $x\wedge y:=\min(x,y)$ and $x\vee y:=\max\{x,y\}$. $\mathcal{P}(\X)$ denotes the set of probability measures on the discrete set $\X$.

\section{Problem Setting}

\subsection{The Main Coding Problems}

Let $\X,\Y,\Z$ be finite alphabets, let $\S$ be another finite set. For every $s\in\S$, let a stochastic matrix 
\[
  W(z\vert x,y\vert s):\quad(x,y,z)\in\X\times\Y\times\Z
\]
be given with inputs from $\X\times\Y$ and outputs from $\Z$. $\S$ is to be interpreted as the set of channel states. We set
\[
  \W :=\{W(\,\cdot\,\vert\,\cdot\,,\,\cdot\,\vert s):s\in\S\}.
\]
We assume that the channel state varies arbitrarily from channel use to channel use. Given words $\x=(x_1,\ldots,x_n)\in\X^n$, $\y=(y_1,\ldots,y_n)\in\Y^n$, and $\z=(z_1,\ldots,z_n)\in\Z^n$, the probability that $\z$ is received upon transmission of $\x$ and $\y$ depends on the sequence $\s\in\S^n$ of channel states attained during the transmission. It equals
\begin{equation}
  W^n(\z\vert\x,\y\vert\s):=\prod_{m=1}^nW(z_m\vert x_m,y_m\vert s_m).
\end{equation}

\begin{defn}
  The set of stochastic matrices
\[
  \{W^n(\,\cdot\,\vert\,\cdot\,,\,\cdot\,\vert\s):\s\in\S^n,n=1,2,\ldots\}
\]
is called the \emph{Arbitrarily Varying Multiple Access Channel (AV-MAC)} determined by $\W $.
\end{defn}

In the traditional non-cooperative encoding schemes used for multiple-access channels, none of the senders has any information about the other sender's message. The goal here is to characterize the capacity region of the AV-MAC achievable when limited information can be exchanged between the encoders. We use Willems conferencing for this exchange \cite{Wi1,Wi2}. If the encoders' message sets are $[1,M_1]$ and $[1,M_2]$, respectively, then this can be described as follows. Let positive integers $V_1$ and $V_2$ be given which can be written as products
\[
  V_\nu=V_{\nu,1}\cdots V_{\nu,I}
\]
for some positive integer $I$ which does not depend on $\nu$. A pair $(c_1,c_2)$ of Willems conferencing functions is determined in an iterative manner via sequences of functions $c_{1,1},\ldots,c_{1,I}$ and $c_{2,1},\ldots,c_{2,I}$. The function $c_{1,i}$ describes what encoder 1 tells the other encoder in the $i$-th conferencing iteration given the knowledge accumulated so far at encoder 1. Thus in general, using the notation 
\[
  \bar\nu:=\begin{cases}
             1&\text{if }\nu=2,\\
             2&\text{if }\nu=1,
           \end{cases}
\]
these functions satisfy for $\nu=1,2$ and $i=2,\ldots,I$
\begin{align*}
  c_{\nu,1}&:[1,M_\nu]\rightarrow[1,V_{\nu,1}],\\
  c_{\nu,i}&:[1,M_\nu]\times[1,V_{\bar\nu,1}]\times\ldots\times[1,V_{\bar\nu,i-1}]\rightarrow[1,V_{\nu,i}].
\end{align*}
These functions recursively define other functions 
\begin{align*}
  c_{\nu,1}^*&:[1,M_\nu]\rightarrow[1,V_{\nu,1}],\\
  c_{\nu,i}^*&:[1,M_1]\times[1,M_2]\rightarrow[1,V_{\nu,i}]
\end{align*}
by
\begin{align*}
  c_{1,1}^*(j)&=c_{1,1}(j),\\
  c_{2,1}^*(k)&=c_{2,1}(k),\\
  c_{1,i}^*(j,k)&=c_{1,i}\bigl(j,c_{2,1}^*(k),\ldots,c_{2,i-1}^*(j,k)\bigr),\\
  c_{2,i}^*(j,k)&=c_{2,i}\bigl(k,c_{1,1}^*(j),\ldots,c_{1,i-1}^*(j,k)\bigr).
\end{align*}
Then we set  
\[
  c_\nu(j,k):=\bigl((c_{\nu,1}^*(j,k),\ldots,c_{\nu,I}^*(j,k)\bigr).
\]
Observe that given a message pair $(j,k)$, the conferencing outcome $(c_1(j,k),c_2(j,k))$ is known at both transmitters. If all conferencing protocols were allowed, the encoders could inform each other precisely about their messages, so this would turn the MAC into a single-sender channel. Thus for nonnegative numbers $C_1,C_2$, if conferencing is used in a blocklength-$n$ code, Willems introduces the restrictions
\begin{equation}\label{confcap}
  \frac{1}{n}\log V_\nu\leq C_\nu.
\end{equation}
$C_1,C_2$ are called the conferencing capacities. Having introduced Willems conferencing, we can now define the codes we are going to consider. 

\begin{defn}\label{defn_code}
\begin{enumerate}
\item Let $n,M_1,M_2$ be positive integers and $C_1,C_2\geq0$. A \emph{deterministic code$_\CONF(n,M_1,M_2,C_1,C_2)$ with blocklength $n$, codelength pair $(M_1,M_2)$, and conferencing capacities $C_1,C_2$} is given by functions $c_1,c_2,f_1,f_2,\Phi$. Here, $(c_1,c_2)$ is a Willems conferencing protocol satisfying \eqref{confcap}. $f_1,f_2$ are the encoding functions
\begin{align*}
  f_1&:[1,M_1]\times[1,V_2]\rightarrow\X^n,\\
  f_2&:[1,M_2]\times[1,V_1]\rightarrow\Y^n.
\end{align*}
The \emph{decoding function} $\Phi$ is a function
\[
  \Phi:\Z^n\rightarrow[1,M_1]\times[1,M_2].
\]
\item A \emph{random code$_\CONF(n,M_1,M_2,C_1,C_2)$ with blocklength $n$, codelength pair $(M_1,M_2)$, and conferencing capacities $C_1,C_2$} is a pair $(C,G)$, where $C=\{C(\gamma):\gamma\in\Gamma\}$ is a finite family of deterministic codes$_\CONF(n,M_1,M_2,C_1,C_2)$, and where $G$ is a random variable taking values in $\Gamma$.
\end{enumerate}
\end{defn}

\begin{figure}
  \begin{center}
    \includegraphics[width=\columnwidth]{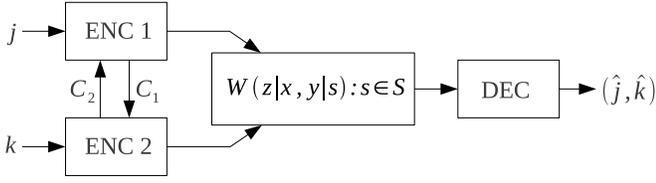}
  \end{center}
  \caption{The AV-MAC $\W $ with conferencing encoders.}\label{fig:confbild}
\end{figure}

Note that a code$_\CONF(n,M_1,M_2,0,0)$ is a traditional MAC code without conferencing. An AV-MAC together with the above coding procedure is called an \emph{AV-MAC with conferencing encoders}, see Fig.\ \ref{fig:confbild}. A code$_\CONF$ $(n,M_1,M_2,C_1,C_2)$ defined by $(c_1,c_2,f_1,f_2,\Phi)$ gives rise to a family 
\begin{equation}\label{codeform}
  \{(\x_{jk},\y_{jk},F_{jk}):(j,k)\in[1,M_1]\times[1,M_2]\},
\end{equation}
where
\begin{align*}
  &\x_{jk}:=f_1(j,c_2(j,k))\in\X^n,\\&\y_{jk}:=f_2(k,c_1(j,k))\in\Y^n,\\&F_{jk}:=\Phi^{-1}\{(j,k)\}\subset\Z^n.
\end{align*}
If the message pair $(j,k)$ is present at the senders, the \emph{codewords} $\x_{jk}$ and $\y_{jk}$ are sent. The decoding sets $\{F_{jk}:(j,k)\in[1,M_1]\times[1,M_2]\}$ give a partition of $\Z^n$ which, just like $\Phi$, assigns to every channel output $\z\in\Z^n$ a message pair which the receiver will decide for upon reception of $\z$.

Note that every family \eqref{codeform}, where the $F_{jk}$ are disjoint, together with a Willems conferencing protocol $(c_1,c_2)$ satisfying \eqref{confcap} defines a code$_\CONF(n,M_1,M_2,C_1,C_2)$ if
\begin{align}
  \x_{jk}&=\x_{jk'}\qquad\text{if }c_2(j,k)=c_2(j,k'),\label{conf1}\\
  \y_{jk}&=\y_{j'k}\qquad\text{if }c_1(j,k)=c_1(j',k).\label{conf2}
\end{align}
Thus a code$_\CONF$ can equivalently be defined by a family \eqref{codeform} together with a conferencing protocol $(c_1,c_2)$ such that \eqref{conf1} and \eqref{conf2} are satisfied. We will often refer to a code$_\CONF$ using the description \eqref{codeform}, and usually without specifying the corresponding conferencing protocol by just assuming that there is one.

The first example of this convention is encountered in our definition of the average error, where the explicit form of the conferencing protocol is irrelevant. 

\begin{defn}
\begin{enumerate}
  \item A code$_\CONF(n,M_1,M_2,C_1,C_2)$ defining a family \eqref{codeform} has an average error probability less than $\lambda\in(0,1)$ if
\[
  \frac{1}{M_1M_2}\sum_{j,k}W^n(F_{jk}^c\vert\x_{jk},\y_{jk}\vert \s)\leq\lambda\quad\text{for all }\s\in\S^n.
\]
  \item Let a random code$_\CONF(n,M_1,M_2,C_1,C_2)$ with the form $(C,G)$ be given. Assume that the deterministic code$_\CONF$ $C(\gamma)$ has the form
\[
  \{(\x_{jk}^\gamma,\y_{jk}^\gamma,F_{jk}^\gamma):(j,k)\in[1,M_1]\times[1,M_2]\}.
\]
Then for any $\s\in\S^n$, define
\begin{equation}\label{P_edefn}
  P_e(C(\gamma)\vert\s):=\frac{1}{M_1M_2}\sum_{j,k}W^n((F_{jk}^\gamma)^c\vert\x_{jk}^\gamma,\y_{jk}^\gamma\vert\s)
\end{equation}
to be the average error incurred by $C(\gamma)$ under channel conditions $\s$. Assume that $G$ has distribution $p_G$. We say that the random code$_\CONF$ defined by $(C,G)$ has an average error smaller than $\lambda\in(0,1)$ if 
\[
  \sum_{\gamma\in\Gamma}P_e(C(\gamma)\vert\s)p_G(\gamma)\leq\lambda\qquad\text{for every }\s\in\S^n. 
\]
\end{enumerate}
\end{defn}
 
This means that uniformly for every interfering sequence, transmission using this code is reliable up to the average error level $\lambda$. The possible state sequences are not weighted by any probability measure. One can interpret this in a communication setting with an adversary who knows which words $\x,\y$ are input into the channel by the senders and then can choose any state sequence $\s\in\S^n$ in order to obstruct the transmission of $\x$ and $\y$. The goal of the encoders then is to enable reliable communication no matter what sequence $\s$ the bad guy might use.

The concept of achievability of a rate pair is the usual one except that conferencing codes$_\CONF$ are allowed for code construction. 

\begin{defn}
  A rate pair $(R_1,R_2)$ is \emph{achievable by the AV-MAC with conferencing encoders and conferencing capacities $C_1,C_2$ under deterministic/random coding} if for every $\lambda\in(0,1)$ and for every $\eps>0$, for $n$ sufficiently large, there is a deterministic/random code$_\CONF(n,M_1,M_2,C_1,C_2)$ with 
\[
  \frac{1}{n}\log M_\nu\geq R_\nu-\eps\quad(\nu=1,2),
\]
and with an average error smaller than $\lambda$. The set of achievable rates under deterministic/random coding is called the \emph{deterministic/random capacity region of the AV-MAC with conferencing encoders and conferencing capacities $C_1,C_2$} and denoted by $\C_d(\S,C_1,C_2)$ (for deterministic coding) and $\C_r(\S,C_1,C_2)$ (for random coding).
\end{defn}

We can now formulate the coding problems which are at the center of this work: 
\begin{center}
  \textit{Characterize the deterministic/random capacity regions $\C_d(\S,C_1,C_2)$ and $\C_r(\S,C_1,C_2)$ of the AV-MAC with conferencing capacities $C_1,C_2$.}
\end{center}

Of course, the main focus is on the deterministic capacity region $\C_d(\S,C_1,C_2)$ as the random capacity region $\C_r(\S,C_1,C_2)$ requires common randomness shared at the encoders and the receiver. For both $\C_d(\S,C_1,C_2)$ and $\C_r(\S,C_1,C_2)$, we need to consider the convex hull $\WQ$ of $\W $. It is parametrized by the set of probability distributions $\P(\S)$ on $\S$, so one can regard $\P(\S)$ as its ``state space''. The stochastic matrix from $\WQ$ assigned to the ``state'' $q\in\P(\S)$ is the matrix with inputs from $\X\times\Y$ and outputs from $\Z$ having the form
\begin{multline*}
  W(z\vert x,y\vert q):=\sum_{s\in\S}W(z\vert x,y\vert s)q(s),\\(x,y,z)\in\X\times\Y\times\Z.
\end{multline*}
We have $\W\subset\WQ$ by identifying $s\in\S$ with the Dirac measure $\delta_s\in\SB$, so that $W(\,\cdot\,\vert\,\cdot\,,\,\cdot\,\vert s)=W(\,\cdot\,\vert\,\cdot\,,\,\cdot\,\vert\delta_s)$. 

Next we define a set of rates $\C^*(\S,C_1,C_2)$. Let $\Pi$ be the set consisting of probability distributions $p\in\P(\U\times\X\times\Y)$, where $\U$ ranges over the finite subsets of the integers and where $p$ has the form
\[
  p(u,x,y)=p_0(u)p_1(x\vert u)p_2(y\vert u).
\]
To each $p\in\Pi$ and $q\in\SB$ one can associate a generic random vector $(U,X,Y,Z_{ q})$ with distribution
\begin{equation}\label{pbars}
 p_{ q}(u,x,y,z)=p(u,x,y)W(z\vert x,y\vert q).
\end{equation}
In this way every $p\in\Pi$ and $ q\in\SB$ define a set $\R(p,q,C_1,C_2)$ consisting of those pairs $(R_1,R_2)$ of nonnegative real numbers which satisfy
\begin{align*}
                R_1&\leq I(Z_{ q};X\vert Y,U)+C_1,\\
                R_2&\leq I(Z_{ q};Y\vert X,U)+C_2,\\
                R_1+R_2&\leq (I(Z_{ q};X,Y\vert U)+C_1+C_2)
                         \wedge I(Z_{ q};X,Y).
\end{align*}
Then set
\[
  \C^*(\S,C_1,C_2):=\bigcup_{p\in\Pi}\;\bigcap_{ q\in\SB}\R(p, q,C_1,C_2).
\]

\begin{thm}\label{thmconfrand}
  For the AV-MAC determined by $\W$ with conferencing capacities $C_1,C_2\geq0$, we have
\[
  \C_r(\S,C_1,C_2)=\C^*(\S,C_1,C_2).
\]
More precisely, for every $(R_1,R_2)\in\C^*(\S,C_1,C_2)$ and every $\eps>0$ there is a $\zeta>0$ and a sequence $(C_n,G_n)$ of random codes$_\CONF(n,M_1^{(n)},M_2^{(n)},C_1,C_2)$ with an average error at most $2^{-n\zeta}$ such that 
        \[
          \frac{1}{n}\log M_\nu^{(n)}\geq R_\nu-\eps\quad(\nu=1,2).
        \]
Additionally the $(C_n,G_n)$ can be chosen such that for every $n$, the constituent deterministic codes$_\CONF$ share the same non-iterative Willems conferencing protocol $(c_1^{(n)},c_2^{(n)})$ given by
\begin{equation}\label{conffunct}
  c_\nu^{(n)}:[1,M_\nu^{(n)}]\rightarrow[1,V_\nu^{(n)}]\qquad(\nu=1,2).
\end{equation}
\end{thm}

\begin{rem}
  The simple form \eqref{conffunct} of conferencing means that no complicated conferencing protocol needs to be designed.
\end{rem}

\begin{rem}
  $\C^*(\S,C_1,C_2)$ was analyzed in \cite{WBBJ11}. It is convex and the auxiliary sets $\U$ can be restricted to have cardinality at most $(\lvert\X\rvert\lvert\Y\rvert+2)\wedge(\lvert\Z\rvert+3)$. Moreover, one can determine finite $C_1,C_2$ such that 
\begin{enumerate}
  \item the full-cooperation sum rate, or 
  \item the full-cooperation capacity region
\end{enumerate}
are achievable. The first statement can be phrased as 
\begin{align*}
  \max_{p\in\Pi}&\min_{q\in\P(\S)}\bigl\{(I(Z_q;X,Y\vert U)+C_1+C_2)\wedge I(Z_q;X,Y)\bigr\}\\&=\max_{p\in\Pi}\min_{q\in\P(\S)}I(Z_q;X,Y)=:C_\infty.
\end{align*}
If for $p\in\Pi$ we set $\Q_p:=\{q\in\P(\S):I(Z_q\wedge X,Y)=C_\infty\}$, then a simple calculation shows that the above condition is satisfied if 
\[
  C_1+C_2\geq C_\infty-\min_{p\in\Pi}\max_{q\in\Q_p}I(Z\wedge X,Y\vert U).
\]
Statement 2) is valid if both $R_1=C_\infty$ and $R_2=C_\infty$ are possible. For this one needs
\begin{align*}
  C_1&\geq C_\infty-\max_{p\in\Pi}\min_{q\in\Q_p}I(Z_q\wedge X\vert Y,U),\\
  C_2&\geq C_\infty-\max_{p\in\Pi}\min_{q\in\Q_p}I(Z_q\wedge Y\vert X,U).
\end{align*}

\end{rem}

\begin{rem}\label{randconv}
  Theorem \ref{thmconfrand} has a weak converse.
\end{rem}

Determining the general deterministic capacity region $\C_d(\S,C_1,C_2)$ is more complex. We give the solution in Theorem \ref{thmconf} below for $C_1\vee C_2>0$. For $C_1=C_2=0$ a partial solution is given in \cite{J},\cite{G},\cite{G2},\cite{AC}. The relation between the cases $C_1=C_2=0$ and $C_1\vee C_2>0$ is discussed in detail in Section \ref{sect:concl}.

\begin{defn}[\cite{G}]\label{defsymm}
  \begin{enumerate}
  \item $\W$ is called \emph{$(\X,\Y)$-symmetrizable} if there is a stochastic matrix
\[
  \sigma(s\vert x,y):\quad(x,y,s)\in\X\times\Y\times\S
\]
such that for every $z\in\Z$ and $x,x'\in\X$ and $y,y'\in\Y$,
\begin{equation*}
  \sum_s\!W(z\vert x,y\vert s)\sigma(s\vert x',y')\!=\!\sum_s\!W(z\vert x',y'\vert s)\sigma(s\vert x,y).
\end{equation*}
  
    \item $\W $ is called \emph{$\X$-symmetrizable} if there is a stochastic matrix
\[
  \sigma_1(s\vert x):\quad(x,s)\in\X\times\S
\]
such that for every $z\in\Z$ and $x,x'\in\X$ and $y\in\Y$,
\[
  \sum_sW(z\vert x,y\vert s)\sigma_1(s\vert x')=\sum_sW(z\vert x',y\vert s)\sigma_1(s\vert x).
\]
  \item $\W $ is called \emph{$\Y$-symmetrizable} if there is a stochastic matrix
\[
  \sigma_2(s\vert y):\quad(y,s)\in\Y\times\S
\]
such that for every $z\in\Z$ and $x\in\X$ and $y,y'\in\Y$,
\[
  \sum_sW(z\vert x,y\vert s)\sigma_2(s\vert y')=\sum_sW(z\vert x,y'\vert s)\sigma_2(s\vert y).
\]
  \end{enumerate}
\end{defn}

\begin{thm}\label{thmconf}
  For the deterministic capacity region of the AV-MAC determined by $\W $ with conferencing capacities $C_1\vee C_2>0$, we have 
\begin{align*}
  \C_d(\S,C_1,C_2)&=\C^*(\S,C_1,C_2)
\intertext{if $\W$ is not $(\X,\Y)$-symmetrizable and}
  \C_d(\S,C_1,C_2)&=\{(0,0)\}
\end{align*}
if $\W$ is $(\X,\Y)$-symmetrizable. As for $\C_r(\S,C_1,C_2)$, the Willems conferencing protocols can be assumed to have the simple non-iterative form \eqref{conffunct}.
\end{thm}

\begin{rem}\label{zweidim}
If $\W $ is not $(\X,\Y)$-symmetrizable, then $\C_d(\S,C_1,C_2)=\C^*(\S,C_1,C_2)$ is  at least one-dimensional. As $C_1\vee C_2>0$, in order to show this it clearly suffices to check that
\begin{equation}\label{verletzt}
  \max_{p\in\Pi}\;\min_{q\in\SB}I(Z_q;X,Y)>0
\end{equation}
if $\W $ is not symmetrizable. If \eqref{verletzt} were violated, then by \cite[Lemma 1.3.2]{CK} there would be a $q\in\SB$ such that
\[
  W(z\vert x,y\vert q)=W(z\vert x',y'\vert q)
\]
for all $x,x'\in\X,y,y'\in\Y,z\in\Z$. Thus $\W $ would be $(\X,\Y)$-symmetrizable using the stochastic matrix
\[
  \sigma(s\vert x,y)=q(s),\quad(x,y,s)\in\X\times\Y\times\S.
\]
But this would contradict our assumption, so \eqref{verletzt} must hold.
\end{rem}

\begin{rem}
  One can regard symmetrizability as the single-letterization of the adversary interpretation of the AV-MAC given above. There, a complete input word pair has to be known to the adversary who can then choose the state sequence. In the definition of $(\X,\Y)$-symmetrizability, the stochastic matrix $\sigma:\X\rightarrow\S$ means that given a \emph{letter} $x\in\X$, the adversary chooses a \emph{random} state $s\in\S$. If $\W$ is $(\X,\Y)$-symmetrizable, the adversary can thus produce a useless single-state MAC $\tilde W:(\X\times\Y)^2\rightarrow\Z$ defined by
\[
  \tilde W(z\vert x,y,x',y')=\sum_{s\in\S}W(z\vert x,y\vert s)\sigma(s\vert x',y').
\]
$\tilde W$ is useless because it is symmetric in $(x,y)$ and $(x',y')$. Thus for word pairs $(\x,\y)$ and $(\x',\y')$, the receiver cannot decide which of the pairs was input into the channel by the senders and which was induced by the adversary's random state choice.
\end{rem}

\begin{rem}
  The above adversary interpretation of symmetrizability makes AV-MACs relevant for information-theoretic secrecy. Clearly, we do not say anything about the decodability of communication taking place in an AV-MAC for non-legitimate listeners. However, reliable communication can be completely prevented in the case the AV-MAC is symmetrizable. A discussion of the single-sender arbitrarily varying wiretap channel can be found in \cite{BBSW}.
\end{rem}

\begin{rem}\label{subexpconf}
  By the definition of Willems conferencing, setting $C_1=C_2=0$ yields the traditional MAC coding, i.e. no conferencing at all is allowed. An inspection of the elimination technique applied in \ref{subsect:randdet} shows that actually it suffices to have conferencing with $V_1=n^2$, so $C_1=(2\log n)/n$ (or, by symmetry, $V_2=n^2$). Using conferencing with this non-constant rate tending to zero in non-$(\X,\Y)$-symmetrizable AV-MACs yields the capacity region $\C^*(\S,0,0)$.
\end{rem}

\begin{rem}\label{detconv}
  If $\W$ is $(\X,\Y)$-symmetrizable, then Theorem \ref{thmconf} almost has a strong converse: it is possible to show that every code that encodes more than one message incurs an average error at least $1/4$. If $\W$ is not $(\X,\Y)$-symmetrizable, then we have a weak converse.
\end{rem}

Theorem \ref{thmconf} does not carry over to the case $C_1=C_2=0$, which is the traditional AV-MAC with non-cooperative coding. To our knowledge, the full characterization of the deterministic capacity region $\C_d(\S,0,0)$ of the AV-MAC without cooperation is still open. We summarize here what has been found out in \cite{AC}, \cite{G2}, \cite{G}, and \cite{J}. For notation, observe that
\begin{align*}
  &\mathrel{\hphantom{=}}\max_{p\in\Pi}\inf_{q\in\SB}I(Z_q;X\vert Y,U)\\
  &=\max_{p\in\Pi}\inf_{q\in\SB}I(Z_q;X\vert Y)\\
  &=\max_{y\in\Y}\max_{r\in\P(\X)}\inf_{q\in\SB}I(Z_q;X\vert Y=y),
\end{align*}
where in the last term, the random vector $(X,Z_q)$ has the distribution $r(x)W(z\vert x,y\vert q)$.

\begin{thm}\label{AC_Satz}
\begin{enumerate}
  \item If $\W $ is neither $(\X,\Y)$- nor $\X$- nor $\Y$-symmetrizable, then $\C_d(\S,0,0)=\C^*(\S,0,0)$ and $\C^*(\S,0,0)$ has nonempty interior. 
  \item If $\W $ is neither $(\X,\Y)$- nor $\X$-symmetrizable, but $\Y$-symmetrizable, then 
\begin{align*}
  &\C_d(\S,0,0)\\&\quad\subset[0,\max_{y\in\Y}\max_{r\in\P(\X)}\inf_{q\in\SB}I(Z_q;X\vert Y=y)]\times\{0\}.
\end{align*}
  \item If $\W $ is neither $(\X,\Y)$- nor $\Y$-symmetrizable, but $\X$-symmetrizable, then 
\begin{align*}
  &\C_d(\S,0,0)\\&\quad\subset\{0\}\times[0,\max_{x\in\X}\max_{r\in\P(\Y)}\inf_{q\in\SB}I(Z_q;Y\vert X=x)].
\end{align*}
  \item If $\W$ is $(\X,\Y)$-symmetrizable, then $\C_d(\S,0,0)=\{(0,0)\}$.
\end{enumerate}
In particular if $\W$ is both $\X$- and $\Y$-symmetrizable, then $\C_d(\S,0,0)=\{(0,0)\}$.
\end{thm}

\begin{rem}
  1) from Theorem \ref{AC_Satz} is due to \cite{AC} and \cite{J}. The other points are due to \cite{G2,G}. The precise characterization of $\C_d(\S,0,0)$ in points 2) and 3) is still open.
\end{rem}

\begin{rem}
  The relation between the three kinds of symmetrizability from Definition \ref{defsymm} is treated in Section \ref{sect:concl}. There we provide the example of an AV-MAC which is both $\X$- and $\Y$-symmetrizable but not $(\X,\Y)$-symmetrizable. 
\end{rem}

\subsection{Related Coding Results}

The set $\WQ$ also determines a compound MAC. This channel differs from the AV-MAC in that it does not change its state during the transmission of a codeword, only constant state sequences are possible. Thus the probability that $\z\in\Z^n$ is received given the transmission of words $\x\in\X^n$, $\y\in\Y^n$ only depends on the given state $ q\in\SB$. It equals
\begin{equation}\label{Wprodmass}
  W^n(\z\vert\x,\y\vert\q_q):=\prod_{m=1}^nW(z_m\vert x_m,y_m\vert q),
\end{equation}
where we denote elements of $\P(\S)^n$ by $\q$ and set $\q_q:=(q,\ldots,q)\in\P(\S)^n$.

\begin{defn}
  The set of stochastic matrices
\[
  \{W^n(\,\cdot\,\vert\,\cdot\,,\,\cdot\,\vert\q_q):q\in\SB,n=1,2,\ldots\}
\]
is called the \emph{compound Multiple Access Channel (compound MAC)} determined by $\WQ$.
\end{defn}

One uses the same deterministic codes$_\CONF$ as for the AV-MAC. Let 
\[
  \{(\x_{jk},\y_{jk},F_{jk}):(j,k)\in[1,M_1]\times[1,M_2]\}
\]
be such a code$_\CONF$. It has an average error less than $\lambda$ for the compound MAC determined by $\WQ$ if
\[
  \frac{1}{M_1M_2}\sum_{j,k}W^n(F_{jk}^c\vert\x_{jk},\y_{jk}\vert\q_q)\leq\lambda
\]
for every $q\in\SB$. Using this average error criterion, the concept of achievability and the definition of the capacity region is analogous to that for deterministic coding for AV-MACs. 

\begin{rem}
Comparing the error criteria for the AV-MAC and the compound MAC from the adversary perspective, one observes that the AV-MAC yields a significantly more robust performance. Theorem \ref{thmconf} describes the region achievable if transmission is reliable for every possible sequence the adversary might choose, whereas Theorem \ref{thmcomp} describes the region which is achievable if the adversary is restricted to constant state sequences.
\end{rem}

In \cite{WBBJ11} the following theorem was proved.

\begin{thm}\label{thmcomp}
  The capacity region of the compound MAC determined by $\WQ$ with conferencing capacities $C_1,C_2\geq 0$ equals $\C^*(\S,C_1,C_2)$. More precisely, for every achievable rate pair $(R_1,R_2)\in\C^*(\S,C_1,C_2)$ and every $\eps>0$, there is a $\zeta>0$ and a sequence of codes$_\CONF$ $(n,M_1^{(n)},M_2^{(n)},C_1,C_2)$ with an average error at most $2^{-n\zeta}$ and
\[
  \frac{1}{n}\log M_\nu^{(n)}\geq R_\nu-\eps\quad(\nu=1,2).
\]
These codes$_\CONF$ can be chosen such that their conferencing protocols have the form \eqref{conffunct}.
\end{thm}

Finally we have to recall the definition and a corollary of the deterministic coding result for single-user Arbitrarily Varying Channels (AVCs). Let $\A$ be a finite input alphabet, $\B$ a finite output alphabet, and $\S$ a finite state set. Let a family 
\[
  \H:=\{H(\,\cdot\,\vert\,\cdot\,\vert s):s\in\S)\}
\]
of stochastic matrices
\[
  H(b\vert a\vert s):\quad(a,b)\in\A\times\B
\]
be given. As for AV-MACs, every state sequence $\s=(s_1,\ldots,s_n)\in\S^n$ determines a new stochastic matrix
\begin{equation*}
  H^n(\b\vert\a\vert\s):=\prod_{m=1}^nH(b_m\vert a_m\vert s_m):\quad(\a,\b)\in\A^n\times\B^n.
\end{equation*}

\begin{defn}
  The set of stochastic matrices 
\[
  \{H^n(\,\cdot\,\vert\,\cdot\,\vert\s):\s\in\S^n,n=1,2,\ldots\}
\]
is called the \emph{Arbitrarily Varying Channel (AVC)} determined by $\H$.
\end{defn}

The admissible codes are classical single-user codes as used for discrete memoryless channels. If such a code with blocklength $n$ and codelength $M$ has the form 
\[
  \{(\a_\ell,F_\ell):\ell\in[1,M]\},
\]
then the average error incurred by this code is smaller than $\lambda\in(0,1)$ if
\[
  \frac{1}{M}\sum_\ell H^n(F_\ell^c\vert\a_\ell\vert\s)\leq\lambda\qquad\text{for all }\s\in\S^n.
\]
Then it is obvious what is meant by ``achievable rates'' and ``capacity'' for $\H$. The capacity of single-user AVCs, which was determined in \cite{CN}, exhibits a dichotomy similar to the one claimed in Theorem \ref{thmconf}. It is described by the original symmetrizability concept from \cite{E}.

\begin{defn}\label{susymm}
  $\H$ is called \emph{symmetrizable} if there is a stochastic matrix
\[
  \sigma(s\vert a):\quad(a,s)\in\A\times\S
\]
such that for every $b\in\B$ and $a,a'\in\A$
\[
  \sum_sH(b\vert a\vert s)\sigma(s\vert a')=\sum_sH(b\vert a'\vert s)\sigma(s\vert a).
\]
\end{defn}

\begin{rem}\label{rem:equi}
  Clearly, the $(\X,\Y)$-symmetrizability of the MAC $\W $ means nothing but symmetrizability of $\W $ when considered as a set of stochastic matrices with inputs from the alphabet $\A=\X\times\Y$.
\end{rem}

\begin{thm}[\cite{CN}, Theorem 1]\label{CN_Satz}
  The deterministic capacity of the single-user AVC determined by $\H$ is positive if and only if $\H$ is not symmetrizable. If $\H$ is symmetrizable, then every code with at least two codewords incurs an average error at least $1/4$.
\end{thm}

\section{The Direct Parts}\label{sect:ach}

We derive the direct part of Theorem \ref{thmconfrand} from Theorem \ref{thmcomp} in Subsections \ref{comparb} and \ref{subsect:amount}. Then, if $\W $ is not $(\X,\Y)$-symmetrizable, we derandomize in Subsections \ref{subsect:posrat} and \ref{subsect:randdet} to obtain the direct part of Theorem \ref{thmconf}.

\subsection{From Compound to Arbitrarily Varying}\label{comparb}

Here we prove the direct part of Theorem \ref{thmconfrand}. We use Ahlswede's ``robustification lemma''. Let $S_n$ be the symmetric group (the group of permutations) on the set $[1,n]$. $S_n$ operates on $\S^n$ by $\pi(\s):=(s_{\pi(1)},\ldots,s_{\pi(n)})$ for any $\pi\in S_n$ and $\s=(s_1,\ldots,s_n)\in\S^n$. Further recall the notation $\q_q$ defined in \eqref{Wprodmass}.

\begin{lem}[\cite{A3}, Lemma RT]\label{robusti}
  If $h:\S^n\rightarrow[0,1]$ satisfies for a $\lambda\in(0,1)$ and for all $q\in\P(\S)$ the inequality
\begin{equation}\label{h-bed}
  \sum_{\s\in\S^n}h(\s)\q_q(\s)\geq1-\lambda,
\end{equation}
then it also satisfies the inequality
\[
  \frac{1}{n!}\sum_{\pi\in S_n}h(\pi(\s))\geq 1-(n+1)^{\lvert\S\rvert}\lambda\qquad\text{for all }\s\in\S^n.
\]
\end{lem}

Now let $(R_1,R_2)\in\C^*(\S,C_1,C_2)$. Theorem \ref{thmcomp} states that for any $\eps>0$ there is a $\zeta>0$ such that for sufficiently large $n$ there is a code$_\CONF(n,M_1,M_2,C_1,C_2)$ with an average error at most $2^{-n\zeta}$ and satisfying
\[
  \frac{1}{n}\log M_\nu\geq R_\nu-\eps\quad(\nu=1,2).
\]
Writing this code$_\CONF$ in the form \eqref{codeform}, this means for every $ q\in\SB$ that
\begin{equation}\label{comperr}
  \frac{1}{M_1M_2}\sum_{j,k}W^n(F_{jk}\vert\x_{jk},\y_{jk}\vert\q_q)\geq 1-2^{-n\zeta}.
\end{equation}
We would like to apply Lemma \ref{robusti} with $\lambda=2^{-n\zeta}$ to the function $h:\S^n\rightarrow[0,1]$ defined by
\[
  h(\s):=\frac{1}{M_1M_2}\sum_{j,k}W^n(F_{jk}\vert\x_{jk},\y_{jk}\vert\s).
\]
Thus we need to show that $h$ satisfies \eqref{h-bed}. Let $q\in\P(\S)$. By \eqref{comperr}, one obtains
\begin{align*}
  &\mathrel{\hphantom{=}}\sum_{\s\in\S^n}h(\s)q^n(\s)\\
  &=\frac{1}{M_1M_2}\sum_{j,k}\sum_{\z\in F_{jk}}\sum_{\s\in\S^n}W^n(\z\vert\x_{jk},\y_{jk}\vert\s)\q_q(\s)\\
  &=\frac{1}{M_1M_2}\sum_{j,k}\sum_{\z\in F_{jk}}W^n(\z\vert\x_{jk},\y_{jk}\vert\q_q)\\
  &\geq 1-2^{-n\zeta},
\end{align*}
and \eqref{h-bed} is satisfied. Applying Lemma \ref{robusti}, one obtains
\begin{equation}\label{randerr}
  \frac{1}{n!}\sum_{\pi\in S_n}h(\pi(\s))\geq 1-(n+1)^{\lvert\S\rvert}2^{-n\zeta}\qquad\text{for all }\s\in\S^n.
\end{equation}
Recall that $\pi^{-1}$ also is an element of $S_n$. Writing $\pi^{-1}(F_{jk})=\{\pi^{-1}(\z):\z\in F_{jk}\}$, the left side of \eqref{randerr} equals
\begin{align}
  &\hphantom{\mathrel{=}}\frac{1}{n!}\sum_{\pi\in S_n}\biggl(\frac{1}{M_1M_2}\sum_{j,k}W^n(F_{jk}\vert\x_{jk},\y_{jk}\vert\pi(\s))\biggr)\notag\\
  &=\frac{1}{n!}\sum_{\pi\in S_n}\biggl(\frac{1}{M_1M_2}\times\notag\\
  &\mathrel{\hphantom{=}}\times\sum_{j,k}W^n\left(\pi^{-1}(F_{jk})\vert\pi^{-1}(\x_{jk}),\pi^{-1}(\y_{jk})\vert\s\right)\biggr).\label{expr}
\end{align}
Because of the bijectivity of $\pi^{-1}$, the family of sets $\{\pi^{-1}(F_{jk}):(j,k)\in[1,M_1]\times[1,M_2]\}$ is disjoint. Thus \eqref{expr} is the average error expression of the random code $(C,G)$ when applied in the AV-MAC determined by $\W$, where $G$ is uniformly distributed on $\Gamma:=S_n$ and where for every $\pi\in S_n$,
\begin{align*}
  C(\pi):=\{\bigl(\pi^{-1}(\x_{jk}),&\pi^{-1}(\y_{jk}),\pi^{-1}(F_{jk})\bigr):\\&(j,k)\in[1,M_1]\times[1,M_2]\}.
\end{align*}
The conferencing protocol remains the same for every $C(\pi)$, as the conference only concerns the messages and not the codewords. By \eqref{randerr} the average error of this random code is less than $(n+1)^{\lvert\S\rvert}2^{-n\zeta}$, hence it tends to zero exponentially. Thus $\C^*(\S,C_1,C_2)\subset\C_r(\S,C_1,C_2)$, which proves the direct part of Theorem \ref{thmconfrand}.

\subsection{Bounding the amount of correlation}\label{subsect:amount}

As a first derandomization step to proving the direct part of Theorem \ref{thmconf}, we have to show the following lemma.

\begin{lem}\label{n^2}
  To every random code$_\CONF(n,M_1,M_2,C_1,C_2)$ with average error at most $\lambda$ which is given as a pair $(C,G)$ there exists a random code$_\CONF(n,M_1,M_2,C_1,C_2)$ with an average error smaller than $3\lambda$ given as a pair $(C',G')$ where $C'\subset C$ and $\lvert C'\rvert= n^2$ and where $G'$ is uniformly distributed on $[1,n^2]$.
\end{lem}

For the proof of Lemma \ref{n^2}, we need a simple result from \cite[Section IV]{J}.

\begin{lem}\label{Jahnslemma}
  Let $N$ i.i.d. random variables $T_1,\ldots,T_N$ with values in $[0,1]$ and underlying probability measure $\PP$ be given. Let $\bar\lambda>0$. Denote by $\EE$ the expectation corresponding to $\PP$. Then
\[
  \PP\left[\frac{1}{N}\sum_{m=1}^NT_m>\bar\lambda\right]\leq\exp\left(-(\bar\lambda-e\EE[T_1])N\right).
\]
\end{lem}

\begin{proof}[Proof of Lemma \ref{n^2}]
Let a random code$_\CONF$ $(C,G)$ with blocklength $n$ and average error smaller than $\lambda$. Recalling our notation \eqref{P_edefn}, the fact that $(C,G)$ has an average error less than $\lambda$ can be stated as
\[
  \EE[P_e(C(G)\vert \s)]\leq\lambda\quad\text{for every }\s\in\S^n.
\]
Let $G_1,\ldots,G_{n^2}$ be independent copies of $G$. This induces a family of $n^2$ independent copies of $(C,G)$. The goal is to show
\begin{equation}\label{goal}
  \PP\biggl[\;\frac{1}{n^2}\sum_{m=1}^{n^2} P_e(C(G_m)\vert \s)\leq3\lambda\text{ for all }\s\in\S^n\biggr]>0.
\end{equation}
Given \eqref{goal}, there is a realization $(\gamma_1,\ldots,\gamma_{n^2})$ of $(G_1,\ldots,G_{n^2})$ such that 
\begin{equation}\label{reform}
  \frac{1}{n^2}\sum_{m=1}^{n^2} P_e(C(\gamma_m)\vert \s)\leq3\lambda
\end{equation}
for every $\s\in\S^n$. Then one defines a random code $(C',G')$ by setting
\[
  C':=\{C(\gamma_m):m \in[1,n^2]\}
\]
and by taking $G'$ to be uniformly distributed on $[1,n^2]$. The expression \eqref{reform} then is nothing but the statement that the average error of the random code $(C',G')$ is smaller than $3\lambda$, and we are done.

It remains to prove \eqref{goal}. $\S$ is finite by assumption, so $\lvert\S^n\rvert$ grows exponentially with blocklength. Hence it suffices to show that
\begin{equation}\label{suffices}
  \PP\biggl[\;\frac{1}{n^2}\sum_m P_e(C(G_m)\vert \s)>3\lambda\biggr]
\end{equation}
is superexponentially small uniformly in $\s\in\S^n$. Let us fix an $\s\in\S^n$. The $G_m$ are i.i.d.\ copies of $G$, so by Lemma \ref{Jahnslemma}, the term \eqref{suffices} is smaller than
\begin{equation}\label{bound}
  \exp\bigl(-\left(3\lambda-e\,\EE[P_e(C(G)\vert \s)]\right)n^2\bigr).
\end{equation}
By assumption
\[
  \EE[P_e(C(G)\vert \s)]\leq\lambda,
\]
so the exponent in \eqref{bound} is negative. This gives the desired superexponential bound on \eqref{suffices}.
\end{proof}

\begin{rem}
  Note that we cannot require the codes$_\CONF$ with at most $n^2$ values of $G$ to have an exponentially small probability of error. This is due to the fact that the exponent in \eqref{bound} must not decrease exponentially in order for the proof to work. Thus there is a trade-off between the error probability and the number of deterministic component codes$_\CONF$ of the random codes$_\CONF$ used to achieve the random capacity region of the AV-MAC with conferencing encoders.
\end{rem}

\subsection{A Positive Rate}\label{subsect:posrat}

In the second derandomization step, we show that if $\W $ is not $(\X,\Y)$-symmetrizable and $C_1>0$ or $C_2>0$ 
then the encoder with the positive conferencing capacity achieves a positive rate by deterministic coding. Without loss of generality we may assume that $C_1>0$.

\begin{thm}\label{thmposrat}
  Let $C_1>0$. If $\W $ is not $(\X,\Y)$-symmetrizable, then there exists an $R$ with $0<R<C_1$ such that the rate pair $(R,0)\in\C^*(\S,C_1,C_2)$ is deterministically achievable using codes$_\CONF$ with conferencing capacity pair $(C_1,0)$ such that the conferencing function $c_1$ is the identity on the message set.
\end{thm}

\begin{proof}
By Remark \ref{rem:equi} and Theorem \ref{CN_Satz}, $\W $ considered as a single-user AVC with vector inputs from the alphabet $\X\times\Y$ has positive capacity. The idea of the proof is to construct from a code for this single-user AVC a code$_\CONF$ such that the first transmitter achieves a positive rate. There is a positive rate $R<C_1$ which is deterministically achievable by the single-user AVC determined by $\W $. This means that for every $\lambda\in(0,1)$ and every $\eps>0$, for $n$ large enough, there is a single-user code
\[
  \{(\x_\ell,\y_\ell,F_\ell):\ell\in[1,M_1]\}
\]
 for $\W $ with
\[
  2^{n(R-\eps/2)}\leq M_1\leq2^{nR}
\]
and with
\[
  \frac{1}{M_1}\sum_{\ell=1}^{M_1}W^n(F_\ell^c\vert\x_\ell,\y_\ell\vert\s)\leq\lambda\qquad\text{for all }\s\in\S^n.
\]

By setting $c_1$ to be the identity on $[1,M_1]$, this code becomes a code$_\CONF(n,M_1,1,C_1,0)$. This is allowed because $\log M_1\leq nR\leq nC_1$. The encoding and decoding functions are defined in the obvious way. Thus the positive rate pair $(R,0)$ is achievable.
\end{proof}

\subsection{From Random to Deterministic}\label{subsect:randdet}

Finally we can show that if $\W $ is not $(\X,\Y)$-symmetrizable, then $\C^*(\S,C_1,C_2)\subset\C_d(\S,C_1,C_2)$. To do so we follow Ahlswede's ``Elimination Technique'' \cite{A1}, whose idea is to use random codes and to replace the randomness needed there by a prefix code with small blocklength which encodes the set of constituent deterministic codes. We again assume that $C_1>0$.

Theorem \ref{thmposrat} implies that there is a $0<R<C_1$ such that for any $\eps\in(0,R)$ and any $\lambda\in(0,1)$, if $n$ is large, there is a code$_\CONF$
\[
  \{(\x_{\gamma}^*,\y_{\gamma}^*,F_{\gamma}^*):\gamma\in[1,n^2]\}
\]
with blocklength $m$,
\begin{equation}\label{m-bed}
  \frac{2}{R}\log n\leq m\leq\frac{2}{R-\eps}\log n,
\end{equation}
with codelength pair $(n^2,0)$ and with average error smaller than $\lambda$. Further, the conferencing function $c_1^*$ is the identity on the set $[1,n^2]$.

For any $0<\delta<C_1$, let $(R_1,R_2)\in\R(p, q,C_1-\delta,C_2)$ for some $p\in\Pi$ and all $q\in\SB$. By Theorem \ref{thmconfrand}, this rate pair is achievable with conferencing capacities $C_1-\delta$ and $C_2$ under random coding. For every $\eps>0$, $\lambda\in(0,1)$, and large $n$ this implies the existence of a random code$_\CONF(n,M_1,M_2,C_1-\delta,C_2)$ defined by a pair $(C,G)$ such that
\[
  \frac{1}{n}\log M_\nu\geq R_\nu-\eps
\]
and with an average error smaller than $\lambda$. Again by Theorem \ref{thmconfrand}, we may further assume that the deterministic component codes$_\CONF$ $C(\gamma)$ of $C$ share the same conferencing protocol $(c_1,c_2)$. As we do not need an exponential decrease of the average error, we may by Lemma \ref{n^2} assume that $G$ is uniformly distributed on $\Gamma=[1,n^2]$. Let every deterministic component code$_\CONF$ $C(\gamma)$ be given as 
\[
  \{(\x_{jk}^{\gamma},\y_{jk}^{\gamma},F_{jk}^{\gamma}):(j,k)\in[1,M_1]\times[1,M_2]\}.
\]

We now construct a deterministic code$_\CONF$ $(\tilde c_1,\tilde c_2,\tilde f_1,\tilde f_2,\tilde\Phi)$ with blocklength $m+n$, message sets $[1,n^2]\times[1,M_1]$ and $[1,M_2]$ (yielding the codelength pair $(n^2M_1,M_2)$), conferencing capacities $C_1,C_2$, and average error smaller than 2$\lambda$. It is defined via concatenation. We define the conferencing functions to be
\begin{align*}
  \tilde c_1(\gamma,j)&:=(\gamma,c_1(j))\in[1,n^2]\times[1,V_1],\\\tilde c_2(k)&:=c_2(k)\in[1,V_2].
\end{align*}
Note that $(\tilde c_1,\tilde c_2)$ has the form \eqref{conffunct}. It is a permissible conferencing protocol if
\[
  \frac{1}{2\,R^{-1}\log n+n}\log n\leq\delta,
\]
because then
\[
  \frac{1}{m+n}\log(n^2V_1)\leq\frac{1}{2\,R^{-1}\log n+n}\log n+\frac{1}{n}\log V_1\leq C_1.
\]

If the encoders have the messages $(\gamma,j)$ and $k$, respectively, they use the codewords
\[
  \bigl(\x^*_{\gamma},\x^{\gamma}_{jk}\bigr)\in\X^{m+n}\quad\text{and}\quad\bigl(\y^*_{\gamma},\y^{\gamma}_{jk}\bigr)\in\Y^{m+n}.
\]
Together with the conferencing protocol $(\tilde c_1,\tilde c_2)$ defined above, this fixes encoding functions $f_1$ and $f_2$, as \eqref{conf1} and \eqref{conf2} are satisfied. The decoding set of the code$_\CONF$ deciding for the pair $\bigl((\gamma,j),k\bigr)$ is defined to be $F^*_{\gamma}\times F_{jk}^{\gamma}\subset\Z^{m+n}$. Thus the deterministic code$_\CONF$ achieving the rate pair $(R,0)$ is used as a prefix code which distinguishes the deterministic component codes$_\CONF$ of the random code$_\CONF$. In this way, derandomization can be seen as a two-step protocol. Setting $a:=2/(R-\eps)$, the rates of the new code are
\[
  \frac{1}{m+n}\log(nM_\nu)
  \geq\frac{1}{a\frac{\log n}{n}+1}\cdot\frac{1}{n}\log M_\nu\geq R_\nu-2\eps,
\]
where the second inequality holds for all $n$ large enough such that 
\[
  \frac{1}{a\frac{\log n}{n}+1}\geq\frac{R_\nu-2\eps}{R_\nu-\eps}.
\]

The randomness of the random code is needed in the estimation of the average error incurred by this coding procedure. Recall Ahlswede's Innerproduct Lemma \cite{A1}:
\begin{lem}\label{innerprod}
  Let $(\alpha_1,\ldots,\alpha_N)$ and $(\beta_1,\ldots,\beta_N)$ be two vectors with $0\leq\alpha_m,\beta_m\leq 1$ for $m=1,\ldots,N$ which for some $\lambda\in(0,1)$ satisfy
\begin{equation}\label{geschbed}
  \frac{1}{N}\sum_{m=1}^N\beta_m\geq 1-\lambda,\qquad\frac{1}{N}\sum_{m=1}^N\alpha_m\geq 1-\lambda,
\end{equation}
then
\[
  \frac{1}{N}\sum_{m=1}^N\alpha_m\beta_m\geq1-2\lambda.
\]
\end{lem}
We use this lemma with $N=n^2$ and replace the index $m$ by $\gamma\in[1,n^2]$. Fix an $\s\in\S^n$ and set
\begin{align*}
  \alpha_{\gamma}&=W^m(F^*_{\gamma}\vert\x^*_{\gamma},\y^*_{\gamma}\vert\s),\\
  \beta_{\gamma}&=\frac{1}{M_1M_2}\sum_{j,k}W^n\bigl(F_{jk}^{\gamma}\vert\x^{\gamma}_{jk},\y^{\gamma}_{jk}\vert\s\bigr).
\end{align*}
Then the conditions in \eqref{geschbed} are satisfied because both the deterministic prefix code $(c_1^*,c_2^*,f_1^*,f_2^*,\Phi^*)$ and the random code $(C,G)$ with constituent codes $(c_1,c_2,f_1^{\gamma},f_2^{\gamma},\Phi^{\gamma})$ have an average error smaller than $\lambda$. Lemma \ref{innerprod} now implies that the code$_\CONF$ $(\tilde c_1,\tilde c_2,\tilde f_1,\tilde f_2,\tilde\Phi)$ constructed above has an average error probability smaller than $2\lambda$. 

This shows that the rate pair $(R_1,R_2)$ is achievable for $\W $ with conferencing capacities $C_1,C_2$. Consequently one obtains
\[
  \bigcup_{\delta>0}\C^*(\S,C_1-\delta,C_2)\subset\C_d(\S,C_1,C_2).
\]
As the capacity region is closed, $\C^*(\S,C_1,C_2)$, which is the closure of the set on the left-hand side, is contained in $\C_d(\S,C_1,C_2)$ as well. This proves the direct part of Theorem \ref{thmconf}.

\section{Converses for the AV-MAC with Conferencing Encoders}\label{sect:conv}

Here we prove the converses claimed in Remark \ref{randconv} and \ref{detconv}. Recall that a weak converse means that, depending on the situation, any deterministic or random code$_\CONF(n,M_1,M_2,C_1,C_2)$ such that the real two-dimensional vector $((1/n)\log M_1,(1/n)\log M_2)$ is at least distance $\eps>0$ from the achievable rate region incurs an average error at least $\lambda(\eps)>0$ if $n$ is large. 

\subsection{Random Coding}

Here we prove the weak converse for Theorem \ref{thmconfrand} (see Remark \ref{randconv}). The idea of the proof is to reduce it to the weak converse for the compound MAC with conferencing encoders defined by $\WQ$ where the use of random codes is allowed. This is proved in the Appendix.

Every $\q\in\P(\S)^n$ induces a product measure via $\q(\s)=q_1(s_1)\cdots q_n(s_n)$. The notation \eqref{Wprodmass} carries over to these general $\q$. Further, recall the notation introduced in \eqref{P_edefn}. We generalize this notation by setting
\[
  P_e(C(\gamma)\vert\q):=\frac{1}{M_1M_2}\sum_{j,k}W^n((F_{jk}^\gamma)^c\vert\x_{jk}^\gamma,\y_{jk}^\gamma\vert\q).
\]
The following lemma is a generalized version of Lemma 2.6.3 in \cite{CK}.
\begin{lem}\label{auxlem}
  For any random code$_\CONF$ which is defined by $(C,G)$ and whose components have the form
\[
  \{(\x_{jk}^\gamma,\y_{jk}^\gamma,F_{jk}^\gamma):(j,k)\in[1,M_1]\times[1,M_2]\},
\]
one has
\begin{align*}
  \sup_{\s\in\S^n}\sum_{\gamma\in\Gamma}P_e(C(\gamma)\vert\s)&p_G(\gamma)\\&=\sup_{\q\in\P(\S)^n}\sum_{\gamma\in\Gamma}P_e(C(\gamma)\vert\q)p_G(\gamma).
\end{align*}
\end{lem}
\begin{proof}
The direction ``$\leq$'' is clear. In order to prove ``$\geq$'', let $\q\in\P(\S)^n$. Clearly
\[
  W^n(\z\vert\x,\y\vert\q)
  =\sum_{\s\in\S^n}\q(\s)W^n(\z\vert\x,\y\vert\s).
\]
Thus
\begin{align*}
  \sum_{\gamma\in\Gamma}P_e(C(\gamma)\vert\q)p_G(\gamma)&=\sum_{\s\in\S^n}\q(\s)\sum_{\gamma\in\Gamma}P_e(C(\gamma)\vert\s)p_G(\gamma)\\
  &\leq\sup_{\s\in\S^n}\sum_{\gamma\in\Gamma}P_e(C(\gamma)\vert\s)p_G(\gamma).
\end{align*}
Upon taking the supremum over $\q\in\P(\S)^n$ on the left-hand side, the lemma is proved.
\end{proof}

Now let a random code$_\CONF$ $(n,M_1,M_2,C_1,C_2)$ be given defined by $(C,G)$ and with average error at most $\lambda$. Assume that the pair $((1/n)\log M_1,(1/n)\log M_2)$ is at distance at least $\eps$ from $\C^*(\S,C_1,C_2)$. Because of Lemma \ref{auxlem}, 
\begin{equation}\label{canonly}
  \lambda_0:=\sup_{q\in\P(\S)}\sum_{\gamma\in\Gamma}P_e(C(\gamma)\vert\q_q)p_G(\gamma)\leq\lambda.
\end{equation}
Thus the random code$_\CONF$ $(C,G)$ has an average error at most $\lambda$ for the compound MAC with conferencing encoders defined by $\WQ$. But the weak converse for the compound MAC with conferencing encoders and random coding, which is proved in the Appendix, implies that \eqref{canonly} can only hold if $\lambda\geq\lambda_0\geq\lambda(\eps)>0$. This concludes the weak converse for the AV-MAC with conferencing encoders using random codes$_\CONF$, and Theorem \ref{thmconfrand} is proved.

\subsection{Deterministic Coding}

\subsubsection{If $\W $ is $(\X,\Y)$-symmetrizable}

If $\W $ is $(\X,\Y)$-symmetrizable, then by Remark \ref{rem:equi} it is also symmetrizable if considered as a single-user AVC with input alphabet $\X\times\Y$. Thus Theorem \ref{CN_Satz} implies any single-user code with at least two codewords incurs an average error greater than $1/4$. Finally, note that every code$_\CONF$ for the AV-MAC with conferencing encoders determined by $\W $ also is a code for the single-user AVC determined by $\W$, so this carries over to the multi-user situation. This proves Theorem \ref{thmconf} if $\W $ is $(\X,\Y)$-symmetrizable.

\subsubsection{If $\W $ is not $(\X,\Y)$-symmetrizable}

We show that the weak converse for the compound MAC determined by $\WQ$ implies the weak converse for the AV-MAC determined by $\W $. Let a code$_\CONF$ $(n,M_1,M_2,C_1,C_2)$ be given. If the rate pair $((1/n)M_1,(1/n)M_2)$ is at least distance $\eps$ away from $\C^*(\S,C_1,C_2)$ and if $n$ is sufficiently large, then there is a $q\in\SB$ such that
\[
  \frac{1}{M_1M_2}\sum_{j,k}W^n(F_{jk}^c\vert\x_{jk},\y_{jk}\vert\q_q)\geq\lambda(\eps)
\]
for some $\lambda(\eps)>0$ because of the weak converse for the compound MAC. Lemma \ref{auxlem} now implies that
\[
  \sup_{\s\in\S^n}\frac{1}{M_1M_2}\sum_{j,k}W^n(F_{jk}^c\vert\x_{jk},\y_{jk}\vert\s)\geq\lambda(\eps)
\]
must hold. Thus the proof of Theorem \ref{thmconf} is complete.

\section{Discussion and Conclusion}\label{sect:concl}

The goal of this paper was to characterize the capacity region of an AV-MAC whose encoders may exchange limited information about their messages. This topic is motivated by the increasing interest of cooperative networks which are subject to exterior interference. For example, spectrum sharing has been discussed for inclusion into future wireless system standards. We saw above that the AV-MAC can be interpreted as a channel suffering from attacks by an adversary who may choose the state sequence given the channel inputs. The reliability requirements for AV-MACs are very strict -- coding is done such that the average error is small for every possible state sequence. The resulting capacity region is the same as that for the conferencing compound MAC determined by the convex hull of the set of channel matrices of the original AV-MAC \textit{if} the latter is not $(\X,\Y)$-symmetrizable. Otherwise the AV-MAC is useless. In contrast, for AV-MACs without conferencing, the complete characterization of the 
deterministic capacity region is still open.

The dichotomy in the form of the deterministic capacity regions of arbitrarily varying multiple-access channels does not occur if random coding is used. However, using a random code requires both the senders and the receiver to have access to a common source of randomness. Thus random coding is usually only used as a mathematical tool for finding good deterministic codes. It is well known that derandomization is no problem for compound channels (including discrete memoryless channels as a special case) -- this is nothing but the well-known ``random coding method''. It builds on the fact that the finite number of channel states does not increase with blocklength. This is not so in the case of arbitrarily varying channels. The number of states \emph{per channel use} remains constant, but the number of states per transmission of a codeword increases \emph{exponentially} in blocklength. This is the reason why the deterministic capacity region of arbitrarily varying multiple-access channels may be strictly 
contained in the random coding capacity region. In fact, if derandomization is not possible, then no positive rates are achievable at all.

In contrast to the derandomization technique used for simpler channels, Ahlswede's elimination technique gives rise to a two-step protocol. In order to approximate a given achievable rate pair $(R_1,R_2)$, one only needs the constituent deterministic codes of a random code whose rate pair approximates $(R_1,R_2)$. The randomness of the random code is used in the average error estimate. (On the other hand, this shows how much weaker the average error criterion is compared with the maximal error requirement -- the randomized part can be ``hidden'' in the average error.)

It is noteworthy that for the arbitrarily varying multiple-access channel, the conferencing protocols needed to achieve any rate pair within the capacity region remain as simple as for the compound multiple-access channel with conferencing encoders. There are no iterative steps, so the implementation of such a conference is straightforward.

Finally we would like to analyze the benefits of Willems conferencing. We compare the gains obtained in AV-MACs to the gains obtained in compound MACs. For both compound and AV-MACs, conferencing may help to achieve positive rates where only the rate pair $(0,0)$ is achievable without transmitter cooperation. This effect is similar to the ``superactivation'' of quantum channels as observed in \cite{SYQuantum}, where it was shown that there are pairs of quantum channels with zero quantum capacity each which achieve positive rates when used together. 

Every compound MAC with conferencing capacities $C_1\vee C_2>0$ and 
\begin{equation}\label{comp>0}
  \max_{p\in\Pi}\min_{q\in\P(\S)}I(Z_q;X,Y)>0
\end{equation}
has an at least one-dimensional capacity region. If \eqref{comp>0} is not satisfied, then the capacity region equals $\{(0,0)\}$. No matter what dimension the corresponding $\C^*(\S,0,0)$ has, the gains of conferencing are continuous in $C_1,C_2$, in particular in $(C_1,C_2)=(0,0)$. This is in contrast to the AV-MAC. The changes in in the deterministic capacity region $\C_d(\S,C_1,C_2)$ are continuous in all $(C_1,C_2)$ with $C_1\vee C_2>0$ because either $\C_d(\S,C_1,C_2)=\C^*(\S,C_1,C_2)$  for all $(C_1,C_2)$ with $C_1\vee C_2>0$ or $\C_d(\S,C_1,C_2)=\{(0,0)\}$ for all $C_1,C_2$. However, there may be a discontinuity in $(C_1,C_2)=(0,0)$.

This corresponds to the two roles conferencing plays in AV-MACs. The ``traditional'' role is to generate a common message and to use the coding result for the (compound) MAC with common message to enlarge the capacity region. For AV-MACs, it does even more -- it changes the channel structure. Recall Remark \ref{subexpconf}. For a conferencing rate pair with $C_1\vee C_2=(2\log n)/n$, the capacity region of the compound MAC stays as it is. Under the conditions that $\W$ is not $(\X,\Y)$-symmetrizable and $\C_d(\S,0,0)\neq\C^*(\S,0,0)$, though, we can strictly enlarge the capacity region of the AV-MAC with this kind of conferencing. 

General conditions for $\C_d(\S,0,0)\neq\C^*(\S,0,0)$ to hold cannot be given because an exact characterization of $\C_d(\S,0,0)$ is generally unavailable. We certainly know by Theorem \ref{AC_Satz} that if $\C_d(\S,0,0)$ is two-dimensional, then $\C_d(\S,0,0)=\C^*(\S,0,0)$. We can further say that if in addition to not being $(\X,\Y)$-symmetrizable, $\W$ is both $\X$- and $\Y$-symmetrizable, then $\C_d(\S,0,0)=\{(0,0)\}$, again by Theorem \ref{AC_Satz}. This is a situation where already the conferencing from Remark \ref{subexpconf} helps. With the same argumentation as in Remark \ref{zweidim} it can easily be seen that 
\[
  \max_{p\in\Pi}\min_{q\in\P(\S)}I(Z_q;X,Y\vert U)>0,
\]
so $\C^*(\S,0,0)$ is at least one-dimensional. Thus there is a discontinuity in $(C_1,C_2)=(0,0)$ in this case. Gubner \cite{G} has found the example of a $\W $ which is both $\X$- and $\Y$-symmetrizable, but not $(\X,\Y)$-symmetrizable.

\begin{ex}
  Let $\X=\Y=\S=\{0,1\}$ and $\Z=\{0,1,2,3\}$. For $s\in\S$ set
\[
  W(z\vert x,y\vert s)=\delta(z-x-y-s),
\]
where $\delta(t)=1$ if $t=0$ and $\delta(t)=0$ else. An equivalent description of this is
\[
  z=x+y+s.
\]
Gubner shows that $\W $ is not $(\X,\Y)$-symmetrizable, but that it is both $\X$- and $\Y$-symmetrizable. Thus this channel is useless if coding is done without conferencing, even though the interfering signal is only added to the sum of the transmitters' signals -- the reliable transmission of messages through the channel is completely prevented. This shows that even the structure of rather simple AV-MACs can be changed by conferencing so as to produce discontinuous jumps at $(C_1,C_2)=(0,0)$.
\end{ex}

\appendix\label{appconv}

Here we prove the weak converse for the compound MAC with conferencing encoders defined by $\WQ$ and with random coding. Let a random code$_\CONF(n,M_1,M_2,C_1,C_2)$ be given which is defined by the pair $(C,G)$. Let this code have average error at most $\lambda$. Denote the conferencing function of the deterministic component code$_\CONF$ with index $\gamma$ by $c_\gamma$. The set $c_\gamma$ maps into is denoted by $[1,V_\gamma]$.

We assume that the pair $((1/n)\log M_1,(1/n)\log M_2)$ is at least distance $\eps$ away from $\C^*(\S,C_1,C_2)$. As all norms are equivalent on the plane, we can without loss of generality work with the $\ell^1$-norm. That means that we assume that
\begin{align*}
  &\sup_{(R_1,R_2)\in\C^*(\S,C_1,C_2)}\!\left\{\left\lvert\frac{1}{n}\log M_1-R_1\right\rvert+\left\lvert\frac{1}{n}\log M_2-R_2\right\rvert\right\}\\&\qquad\quad\geq\eps.
\end{align*}
This statement is equivalent to the fact that for every $p\in\Pi$ there is some $q\in\SB$ such that one of the following inequalities holds:
\begin{align}
  \frac{1}{n}\log M_1&\geq C_1+I(Z_q; X\vert Y,U)+\eps,\label{erstegute}\\
  \frac{1}{n}\log M_2&\geq C_2+I(Z_q; Y\vert X,U)+\eps,\\
  \frac{1}{n}\log M_1M_2\\\geq\{C_1+&C_2+I(Z_q;X,Y\vert U)\wedge I(Z_q;X,Y\vert U)\}+\eps.\label{letztegute}
\end{align}

Our goal is to mainly use arguments already known from the weak converse for deterministic coding, so that we can refer to \cite{WBBJ11}. From the random code$_\CONF$, we define several random variables in addition to $G$:
\begin{itemize}
  \item the pair $(T_1,T_2)$, which is uniformly distributed on $[1,M_1]\times[1,M_2]$ and independent of $G$,\smallskip
  \item the pair $(\tilde U_1,\tilde U_2):=(c_1^G(T_1,T_2),c_2^G(T_1,T_2))$ taking values in $[1,V_1]\times[1,V_2]$, where for $\nu=1,2$ we define $V_\nu=\max_{\gamma\in\Gamma}V_\nu^\gamma$,\smallskip
  \item $\tilde X:=\x_{T_1T_2}^G$, $\tilde Y:=\y_{T_1T_2}^G$,\smallskip
  \item a random variable $\tilde Z\in\Z^n$ which satisfies
\begin{align*}
  &\mathrel{\hphantom{=}}\PP[\tilde Z=\z\vert\tilde X=\x,\tilde Y=\y,\tilde U_1=v_1,\tilde U_2=v_2,\\
  &\qquad\qquad\qquad T_1=j,T_2=k,G=\gamma]\\
  &=W^n(\z\vert\x,\y).
\end{align*}
\end{itemize}

Every $\gamma\in\Gamma$ corresponds to a deterministic code $C(\gamma)$ with average error at most $\lambda_\gamma$. For each of these codes, we can proceed as in \cite{WBBJ11}. That means that we first apply Fano's inequality and then obtain single-letter bounds on the code rates. More precisely, writing $\U=[1,n]\times[1,V_1]\times[1,V_2]$, we can construct for each $\gamma$ a probability distribution $p(u,x,y\vert \gamma)$ on $\U\times\X\times\Y$ which is contained in $\Pi$. This is due to the fact proved in \cite{Wi1} that conditional on $\gamma$ and $(\tilde U_1,\tilde U_2)$, the random variables $\tilde X$ and $\tilde Y$ are independent. Thus we have
\[
  p(u,x,y\vert\gamma)=p_0(u\vert\gamma)p_1(x\vert u,\gamma)p_2(y\vert u,\gamma).
\]
Further, for each $q\in\SB$, we construct the random vector $(U,X,Y,Z_q)$ which together with $G$ has the distribution
\begin{align}\label{Zgut}
  &\mathrel{\hphantom{=}}\PP[Z_q=z,Y=y,X=x,U=u,G=\gamma]\notag\\
  &=W(z\vert x,y\vert q)p(u,x,y\vert\gamma)p_G(\gamma).
\end{align}
By construction, this random vector satisfies for every $q\in\SB$
\begin{align*}
  \frac{1}{n}\log M_1&\leq C_1+I(Z_q;X\vert Y,U,G=\gamma)+\frac{1}{n}\Delta_\gamma,\\
  \frac{1}{n}\log M_2&\leq C_2+I(Z_q;Y\vert X,U,G=\gamma)+\frac{1}{n}\Delta_\gamma,\\
  \frac{1}{n}\log M_1M_2&\leq\{(C_1+C_2+I(Z_q;X,Y\vert U,G=\gamma))\\&\qquad\qquad\wedge I(Z_q;X,Y\vert G=\gamma)\}+\frac{1}{n}\Delta_\gamma,
\end{align*}
where 
\[
  \Delta_\gamma:=2h(2\lambda_\gamma)+4\lambda_\gamma\log M_1M_2.
\]
Next we take the expectation over $G$. Using the concavity of $h$ and \eqref{Zgut}, this can be transformed into
\begin{align}
  \frac{1}{n}\log M_1&\leq C_1+I(Z_q;X\vert Y, U,G)+\frac{1}{n}\Delta,\\
  \frac{1}{n}\log M_2&\leq C_2+I(Z_q;Y\vert X, U,G)+\frac{1}{n}\Delta,\\
  \frac{1}{n}\log M_1M_2&\leq\{(C_1+C_2+I(Z_q;X,Y\vert U,G))\notag\\&\qquad\qquad\wedge I(Z_q;X,Y\vert G)\}+\frac{1}{n}\Delta,\label{wrgstf}
\end{align}
with 
\[
  \Delta:=2h(2\lambda)+4\lambda\log M_1M_2.
\]
As $I(Z_q;X,Y\vert G)=H(Z_q\vert G)-H(Z_q\vert X,Y)$, the concavity of entropy implies that the bound in \eqref{wrgstf} is relaxed if one replaces $I(Z_q;X,Y\vert G)$ by $I(Z_q;X,Y)$. We now set $\hat U:=(U,G)$ and observe that the distribution of $(\hat U,X,Y)$ is contained in $\Pi$. Comparing the resulting set of inequalities with a valid one among \eqref{erstegute}-\eqref{letztegute} and using the same simple arguments as in \cite{WBBJ11}, we can now show that $\lambda\geq\lambda(\eps)>0$. This finishes the proof of the weak converse.

\section*{Acknowledgment}

The authors would like to thank Frans Willems for the fruitful discussions about this work at the Banff workshop ``Interactive Information Theory'' in January 2012. They would also like to thank the associate editor Yossef Steinberg 
for his valuable comments given during the review process of the paper.


\begin{biography}{Moritz Wiese}
  (S'09) received the Dipl.-Math. degree in mathematics from the university of Bonn, Germany, in 2007. He has been pursuing the PhD degree since then. From 2007 to 2010, he was a research assistant at the Heinrich-Hertz-Lehrstuhl f\"ur Mobilkommunikation, Technische Universit\"at Berlin, Germany. Since 2010, he is a research and teaching assistant at the Lehrstuhl f\"ur Theoretische Informationstechnik, Technische Universit\"at M\"unchen, Munich, Germany.
\end{biography}

\begin{biography}{Holger Boche}
(M'04-SM'07-F'11) received the Dipl.-Ing. and Dr.-Ing. degrees in electrical engineering from the Technische Universit\"at Dresden, Dresden, Germany, in 1990 and 1994, respectively. He graduated in mathematics from the Technische Universit\"at Dresden in 1992. From 1994 to 1997, he did postgraduate studies in mathematics at the Friedrich-Schiller Universit\"at Jena, Jena, Germany. He received his Dr. rer. nat. degree in pure mathematics from the Technische Universit\"at Berlin, Berlin, Germany, in 1998. In 1997, he joined the Heinrich-Hertz-Institut (HHI) f\"ur Nachrichtentechnik Berlin, Berlin, Germany. Starting in 2002, he was a Full Professor for mobile communication networks with the Institute for Communications Systems, Technische Universit\"at Berlin. In 2003, he became Director of the Fraunhofer German-Sino Lab for Mobile Communications, Berlin, Germany, and in 2004 he became the Director of the Fraunhofer Institute for Telecommunications (HHI), Berlin, Germany. Since October 2010 he has been with the 
Institute of Theoretical Information Technology and Full Professor at the Technische Universit\"at M\"unchen, Munich, Germany. He was a Visiting Professor with the ETH Zurich, Zurich, Switzerland, during the 2004 and 2006 Winter terms, and with KTH Stockholm, Stockholm, Sweden, during the 2005 Summer term. Prof. Boche is a Member of IEEE Signal Processing Society SPCOM and SPTM Technical Committee. He was elected a Member of the German Academy of Sciences (Leopoldina) in 2008 and of the Berlin Brandenburg Academy of Sciences and Humanities in 2009. He received the Research Award "Technische Kommunikation" from the Alcatel SEL Foundation in October 2003, the "Innovation Award" from the Vodafone Foundation in June 2006, and the Gottfried Wilhelm Leibniz Prize from the Deutsche Forschungsgemeinschaft (German Research Foundation) in 2008. He was co-recipient of the 2006 IEEE Signal Processing Society Best Paper Award and recipient of the 2007 IEEE Signal Processing Society Best Paper Award.
\end{biography}

\end{document}